\documentclass[submission,copyright,creativecommons]{eptcs}
\providecommand{\event}{QPL 2022} 

\usepackage{amsmath,amsfonts,amsthm}
\usepackage{tikz}
\usepackage{tikz-cd}
\usepackage{todonotes}
\DeclareMathAlphabet{\mathcal}{OMS}{cmsy}{m}{n} 

\theoremstyle{plain}
\newtheorem{thm}{Theorem}[section]
\newtheorem{cor}[thm]{Corollary}

\newtheorem{lma}[thm]{Lemma}

\theoremstyle{definition}
\newtheorem{defn}[thm]{Definition}
\newtheorem{exmp}[thm]{Example}

\newtheorem*{lma*}{Lemma}

\theoremstyle{remark}
\newtheorem*{rmrk*}{Remark}

\newcommand{\hide}[1]{}

\newcommand{\Cs}{                   
    \mathbf{C}^\ast
    }
    
\newcommand{\MeasPower}[2]{         
    #1^#2                           
    }

\newcommand{\iid}[1]{               
    \mathrm{iid}_{#1}
    }

\newcommand{\borel}[1]{
    \mathit{Borel} \left( #1 \right) 
    }

\newcommand{\Tr}[1]{
    \mathrm{Tr}\left( #1 \right)
    }
\newcommand{\innerprod}[2]{
    \left\langle #1 \, \middle\vert\, #2 \right\rangle
    }
\newcommand{\outerprod}[2]{
    \left\vert #1 \middle\rangle \middle\langle #2 \right\vert
    }

\newcommand{\sym}[1]{               
    \mathcal S_#1
    }

\newcommand{\op}[1]{
    \left( #1 \right)^{\mathrm{op}}
    }

\newcommand{\Set}{
    \mathbf{Set}
    }
\newcommand{\CH}{
    \mathbf{CH}
    }
\newcommand{\Top}{
    \mathbf{Top}
    }

\newcommand{\ConCom}{
    \mathbf{ConvCH}
    }

\newcommand{\Iinj}{
    \mathbf{I}_{\mathrm{inj}}
    }

\newcommand{\CstMIU}{
    \Cs_{\mathrm{Mult}}
    }
\newcommand{\CCstMIU}{
    \mathbf{c}\Cs_{\mathrm{Mult}}
    }
\newcommand{\CstPU}{
    \Cs_{\mathrm{Pos}}
    }
\newcommand{\CCstPU}{
    \mathbf{c}\Cs_{\mathrm{Pos}}
    }
\newcommand{\CstCPU}{
    \Cs_{\mathrm{CPU}}
    }
    
\newcommand{\kl}[1]{                   
    \mathcal{K}l \left( #1 \right)
    }
\newcommand{\klto}{
    \rightsquigarrow%
    }
\newcommand{\Alg}[1]{                  
    \mathcal{E}m \left( #1 \right)
    }
    
\newcommand{\radu}{                 
    \mathcal{R}
    }
\newcommand{\rad}[1]{
    \radu \left( #1 \right)
    }
\newcommand{\radsqr}[1]{
    \radu^2 \left( #1 \right)
    }

\newcommand{\statu}{                
    S
    }
\newcommand{\stat}[1]{
    \statu \left( #1 \right)
    }
\newcommand{\sstatu}{               
    \mathcal I
    }
\newcommand{\sstat}[1]{
    \sstatu \left( #1 \right)
    }

\newcommand{\cont}[1]{
    C \left( #1 \right)
    }

\newcommand{\CStA}{\mathcal{A}}
\newcommand{\CStB}{\mathcal{B}}  
\newcommand{\CStElA}{a}  

\newcommand{\spotimes}{
\mathop{\otimes_{\mathrm{min}}}    }
\newcommand{\Xon}[1]{
    \CStA^{\otimes #1}
    }
\newcommand{\Xoinf}{
    \CStA^{\otimes \infty}
    }
    
\title{Quantum de Finetti Theorems as Categorical Limits, \\ and Limits of State Spaces of C*-algebras}
\author{Sam Staton \qquad \qquad Ned Summers
\institute{Department of Computer Science,\\
University of Oxford\\
Oxford, United Kingdom}
}

\def\titlerunning{Quantum de Finetti Theorems as Limits}
\def\authorrunning{S. Staton \& N. Summers}
\begin{document}

\maketitle


\begin{abstract}
    De Finetti theorems tell us that if we expect the likelihood of outcomes to be independent of their order, then these sequences of outcomes could be equivalently generated by drawing an experiment at random from a distribution, and repeating it over and over. In particular, the quantum de Finetti theorem says that exchangeable sequences of quantum states are always represented by distributions over a single state produced over and over. The main result of this paper is that this quantum de Finetti construction has a universal property as a categorical limit. This allows us to pass canonically between categorical treatments of finite dimensional quantum theory and the infinite dimensional. The treatment here is through understanding properties of (co)limits with respect to the contravariant functor which takes a C*-algebra describing a physical system to its convex, compact space of states, and through discussion of the Radon probability monad. We also show that the same categorical analysis also justifies a continuous de Finetti theorem for classical probability.
\end{abstract}
\section{Introduction}
\newcommand{\HH}{\mathcal{H}}
\newcommand{\KK}{\mathcal{K}}
\newcommand{\CC}{\mathbb{C}}
\newcommand{\BB}{\mathcal{B}}
\newcommand{\QDF}{\mathit{QdF}}
\usetikzlibrary{calc}
\tikzset{curve/.style={settings={#1},to path={(\tikztostart)
    .. controls ($(\tikztostart)!\pv{pos}!(\tikztotarget)!\pv{height}!270:(\tikztotarget)$)
    and ($(\tikztostart)!1-\pv{pos}!(\tikztotarget)!\pv{height}!270:(\tikztotarget)$)
    .. (\tikztotarget)\tikztonodes}},
    settings/.code={\tikzset{quiver/.cd,#1}
        \def\pv##1{\pgfkeysvalueof{/tikz/quiver/##1}}},
    quiver/.cd,pos/.initial=0.35,height/.initial=0}

The quantum analogue of de Finetti's theorem~\cite{CavesFuchsSchack2001,Hudson1976, Hudson1981, Stormer1969} explains that a ``belief about a quantum state'' has a more elementary description as an exchangeable sequences of quantum states. The point of this paper is to show that this de Finetti theorem can be phrased in categorical terms. Thus we connect this theorem, which is a fundamental theorem of quantum Bayesianism (e.g.~\cite{FuchsSchack2004}), with categorical and compositional approaches to axiomatizations and reconstructions of quantum theory (e.g.~\cite{CoeckeFritzSpekkens2016,HeunenKissingerSelinger2014,Huot2018,Parzygnat2020,Selby2021,Tull2020,Wetering2019}).

Let $\HH $ be a Hilbert space (e.g.~$\CC^2$).
An \emph{sequence} of states on $\HH$ is a collection of quantum states on tensor powers of $\HH$:
\[\text{a state $\rho_0$ for $\HH^{\otimes 0}=\CC $,\ \ 
a state $\rho_1$ for $\HH^{\otimes 1}=\HH$,\ \ 
a state $\rho_2$ for $\HH^{\otimes 2}=\HH \otimes \HH $,\ \ 
a state $\rho_3$ for $\HH^{\otimes 3}, \ \ $\dots}\] 
For example, if $\HH=\CC^2$, then a sequence of states on $\HH$ is a sequence of density matrices in $\CC^{(2^n)^2}$.
A sequence is \emph{exchangeable} if each state commutes with reindexing, e.g.~$\rho_2=\rho_2 \circ \mathrm{swap}$ for the swap map $\HH \otimes \HH \to \HH\otimes\HH$, and if taking the partial trace of $\rho_m$ over any $m-n$ indices gives $\rho_n$,  for $n \leq m$. 
We can phrase this in categorical terms by recalling that a state is a quantum channel $\CC\to \HH$
(i.e.~a CPTP map between the corresponding spaces of density matrices), and so an exchangeable sequence is a commuting cone of quantum channels
\newcommand{\tikzqdf}[2]{
          #1^{\otimes 0} \arrow[loop,distance=1.5em,out=120, in=60]
      \& #1^{\otimes 1}\arrow[loop,distance=1.5em,out=120, in=60] \ar[l]
      \& #1^{\otimes 2}
      \ar[l,shift left = 1.5]
      \ar[l,shift right = 1.5]
      \arrow[loop,distance=1.5em,out=120, in=60]
      \arrow[loop,distance=1.5em,out=-120, in=-70]
      \& #1^{\otimes 3}
      \ar[l,shift left = 2.5]
      \ar[l,shift left = 1.5]
      \ar[l,shift left = 0.5]
      \ar[l,shift right = 0.5]
      \ar[l,shift right = 1.5]
      \ar[l,shift right = 2.5]
      \arrow[loop,distance=1.5em,out=140, in=125]
      \arrow[loop,distance=1.5em,out=115, in=100]      
      \arrow[loop,distance=1.5em,out=90, in=75]      
      \arrow[loop,distance=1.5em,out=65, in=50]      
      \arrow[loop,distance=1.5em,out=-130, in=-110]
      \arrow[loop,distance=1.5em,out=-100, in=-80]      
      \&\cdots 
     \ar[l,shift left = 2.5,"\vdots" {yshift=7mm}]
     \ar[l,shift right = 2.5]
     \\\\ \& \& \& \& \& #2
    \ar[llllluu,curve={height=-30pt}] 
    \ar[lllluu,curve={height=-25pt}] 
    \ar[llluu,curve={height=-20pt}] 
    \ar["\dots"',lluu,curve={height=-15pt}] 
}

\begin{equation}\label{dgm:qdfintro} 
  \begin{tikzcd}[ampersand replacement=\&,row sep=small]
    \tikzqdf{\HH}{\CC}
    \end{tikzcd}
     \end{equation}
Our categorical statement of the quantum de Finetti theorem (Theorem~\ref{thm:qdf}, dgm.~\eqref{dgm:introqlimitmed}) is about a limit for this diagram, i.e.~a universal exchangeable sequence of quantum channels. This follows recent categorical treatments of the classical case~\cite{FritzGondaPerrone2021,JacobsStaton2020}.

To give this categorical statement precisely, we make three steps.
\begin{enumerate}
    \item Firstly, we extend our set up to allow for channels that have both quantum and classical information. Formally, this is done by recalling that the dual category $\op\CstCPU$ of C*-algebras and completely positive unital maps fully embeds the category of quantum channels, but also fully embeds a good deal of classical probability, in terms of Radon probability kernels between compact Hausdorff spaces.
    \begin{equation}\label{eqn:intro:cstar}
    \begin{tikzcd}[row sep=tiny]
        \text{(quantum channels)}\ar[r,hook]&\op\CstCPU&\ar[l,hook']\text{(classical probability kernels)}
    \end{tikzcd}
    \end{equation}
    This move is important. It turns out that there is no Hilbert space that is the limit of diagram~\ref{dgm:qdfintro}. Instead, some classical probability is necessary.
    
    \item Then, rather than look only at exchangeable sequences of states $\CC\to \HH^{\otimes n}$, we look more generally at \emph{parameterized} exchangeable sequences, i.e.~sequences of channels $\KK\to\HH^{\otimes n}$, incorporating both classical and quantum randomness.
    \begin{equation}\label{dgm:catqdfintro}
    \begin{tikzcd}[ampersand replacement=\&,row sep=tiny]
        \tikzqdf{\HH}{\KK}
    \end{tikzcd}
    \end{equation}
    
    \item \textbf{Theorem~\ref{thm:qdf}} (paraphrased): There is a $\Cs$-algebra $\QDF(\HH)$ and a cone
    
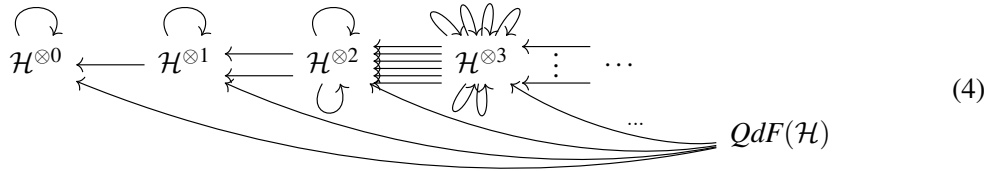
\begin{equation}\label{dgm:introqlimit} \begin{tikzcd}[ampersand replacement=\&,row sep=tiny]
        \tikzqdf{\HH}{\QDF(\HH)}
    \end{tikzcd}\end{equation}
    which is limiting in the category $\op\CstCPU$.\\
    This universal property says that any cone factors uniquely through $\QDF(\HH)$.
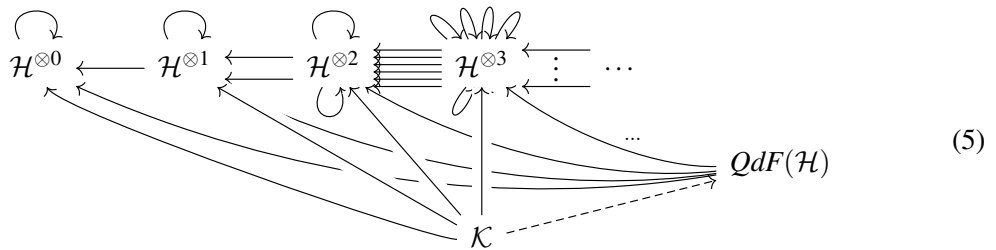
    \begin{equation}\label{dgm:introqlimitmed}  \begin{tikzcd}[ampersand replacement=\&,row sep=small]
       \tikzqdf{\HH}{\QDF(\HH)}
       \\\&\&\&\KK
       \arrow[to=1-1,from=4-4,in=-60,out=180,distance=1em]
       \arrow[to=1-2,from=4-4,crossing over]
       \arrow[to=1-3,from=4-4, crossing over]
       \arrow[to=1-4,from=4-4,crossing over]
       \arrow[to=3-6,from=4-4,dashed]
    \end{tikzcd}\end{equation}
    In other words, to give a cone, i.e.~a parameterised exchangeable sequence of states in $\HH$, is equivalent to giving a channel of quantum and classical information to $\QDF(\HH)$.
\end{enumerate}

The starting point for our proof of Theorem~\ref{thm:qdf} is St\o rmer's quantum de Finetti result~\cite{Stormer1969}.
Inspired by St\o rmer's result, we take the candidate limiting cone $\QDF(\HH)$ to be the $\Cs$-algebra corresponding to the space of classical distributions on \emph{all} states of $\HH$.
Although St\o rmer's work is not phrased in categorical terms at all, it follows from that result that there is a unique mediating morphism in diagram~\eqref{dgm:introqlimitmed} in the case where $\KK=\CC$.
To show that this is indeed is a categorical limit, we follow the following steps.
\begin{itemize}
    \item We note that the candidate limiting cone $\QDF(\HH)$ is classical, that is to say, it lies on the right hand, classical side of \eqref{eqn:intro:cstar}, even though the diagram itself typically lies on the left hand, quantum side of \eqref{eqn:intro:cstar}. So we can consider the categorical limit and diagram~\eqref{dgm:introqlimitmed} in the category of $\Cs$-algebras and \emph{positive} unital maps since positive maps into and out of commutative $\Cs$-algebras are necessarily \emph{completely} positive.
    \item The category of C*-algebras and positive unital maps dually embeds into a category of compact spaces with convex structure, by regarding their states (\S\ref{sec:prelim}, Thm.~\ref{thm: State Space Functor is Full and Faithful}; \cite{Furber2015}). This category of convex compact  spaces is categorically well-behaved, since it is the category of algebras for a monad on another well-behaved category. In this larger category we are able to use standard monadicity results to translate `pointwise' limiting structure to categorical limiting structure (\S\ref{sec:lemmas}, Theorem~\ref{thm: CCLcvx Pointwise Limit}). 
\end{itemize}
In particular, we can then show that diagram~\eqref{dgm:introqlimit} is a categorical limit (\S\ref{sec:qdf}, Theorem~\ref{thm:qdf}).

In this way, we can understand ``belief about a quantum state'' in categorical terms. This opens the door to using categorical diagrammatic notation, which we illustrate in Section~\ref{sec:illustration}. 

\section{Preliminaries}\label{sec:prelim}
We recall rudiments of probability theory (\S\ref{sec:probthy}) and $\Cs$-algebras (\S\ref{sec:cstar},\ref{sec:tensors}). In this context we recall classical and quantum de Finetti theorems (\S\ref{sec:classical-df},\ref{sec:prelim-qdf}) and the $\Cs$-algebraic treatment of probability  (\S\ref{sec:probgelfand},\ref{sec:states}).

            \subsection{Rudiments of Probability Theory}
            \label{sec:probthy}
            We begin by recalling some measure-theoretic probability theory. 
            
            \begin{defn}[Probability Measure]\label{defn:probmeasure}
            For a set $X$, \textit{a $\sigma$-algebra for $X$} is a collection of subsets of $X$, $\Sigma_X\subseteq \mathcal P (X)$, which contains $X$ and is closed under countable unions and complements.
            
            \textit{A measurable space} is a pair $(X,\Sigma_X)$ of a set $X$ and a $\sigma$-algebra $\Sigma_X$ on $X$. 
            
            In what follows, we are almost exclusively concerned with measures on topological spaces. The \textit{Borel $\sigma$-algebra $\borel X$ on a topological space $X$} is the smallest $\sigma$-algebra generated by the open sets of $X$. 
            Additionally, when we refer to finite or countable sets, we consider them as measurable spaces with the $\sigma$-algebra of all possible subsets.

            \textit{A measurable function} between measurable spaces $(X,\Sigma_X)$ and $(Y,\Sigma_Y)$ is a function $f\colon X \to Y$ such that for any $S \in \Sigma_Y$, we have $f^{-1}\left( S\right) \in \Sigma_X$. Continuous functions between topological spaces $(X, \borel{X}) \to (Y, \borel{Y})$ are measurable, though not all measurable functions are continuous. 
            
            \textit{A probability measure on a measurable space $(X, \Sigma_X)$} is a function $\mu \colon \Sigma_X \to [0,1]$ such that $\mu(X) = 1$ and for a disjoint countable collection of sets $\{U_i\}_{i \in N}\subset\Sigma_X$, $\mu \left(\bigcup_{i \in \mathbb N} U_i \right) = \sum_{i\in \mathbb N} \mu(U_i)$. Given a probability measure $\mu \colon \Sigma_X \to [0,1]$ and a measurable function $f\colon X \to Y$, the \textit{pushforward of $\mu$ by $f$}, a probability measure on $Y$, is denoted $f_\ast \mu$ and given by $f_\ast \mu(S) := \mu \left( f^{-1}(S) \right)$ for measurable $S \in \Sigma_Y$.
            
            All our topological spaces will be compact, Hausdorff spaces, so we will only use measures that behave well with the compactness of the space. Let $X$ be a compact Hausdorff space. \textit{A Radon probability measure on $X$}, $\mu \colon \borel{X} \to [0,1]$, is a probability measure on the Borel $\sigma$-algebra which is \textit{inner regular}: for any measurable $S \subseteq X$, $\mu(S) = \sup_{K \subseteq S} \mu (K)$ where $K$ varies over all compact subsets. For any continuous function $f\colon X\to \CC$, a Radon probability measure $\mu$ induces an integral
            $\int_{x\in X} f(x) \,\mathrm{d}\mu\in \CC$, so that we can regard $\mu$ as a map from the set (in fact, space) of continuous functions on $X$ to $\CC$ (see Thm.~\ref{thm: Stochastic Gelfand Duality},  \S\ref{sec:states}).
            \end{defn}
            
        \subsection{Kolmogorov Extension Theorem and Hewitt-Savage de Finetti Theorem}
            \label{sec:classical-df}
            Kolmogorov's extension theorem connects measures on infinite product spaces with measures on finite truncations. Recall that for a set of topological spaces $\{ X_i\}_{i\in I}$, their product has underlying set $\prod_{i \in I} X_i$ and has topology generated by the cylinder sets $V[U_{i_1},\dots,U_{i_n}] = \left\{ (x_i)_{i\in I} \in \prod_{i \in I} X_i \, \middle\vert \, x_{i_k} \in U_{i_k} \text{ for } 1 \leq k \leq n\right\}$, varying over all finite subsets $\{ i_1, \dots, i_n \} \subset I$ and open sets $U_{i_k} \subset X_{i_k}$.
            
            \begin{thm}[Kolmogorov Extension Theorem (e.g. \cite{Okada1978})] \label{thm: Kolmogorov Extension Theorem}
                Let $X$ be a compact Hausdorff space. Let $X^\mathbb N$ be the countable product of copies of $X$. For each finite $N \subset \mathbb N$, let $\mu_N$ be a Radon probability measure on $X^N$ with the product topology, with the condition that if we take finite subsets $M \subset N \subset \mathbb N$, $\mu_M$ is the pushforward of $\mu_N$ by the projection $X^N \to X^M$. Then there exists a unique probability measure $\mu$ on $X^\mathbb N$ such that, for any finite $N \subset \mathbb N$, $\mu_N$ is the pushforward of $\mu$ by the projection $X^\mathbb N \to X^N$. Further, this measure is itself Radon.
            \end{thm}

            From here we can now define an exchangeable measure.
            
            \begin{defn}[Exchangeable Measure]\label{defn: Exchangeable Measure}
            Let $X$ be a measurable space. Let $\mu$ be a measure on $X^\mathbb N$. For each permutation $\sigma \colon \mathbb N \to \mathbb N$, there is an isomorphism 
            \begin{align*}
                \eta_\sigma \colon X^\mathbb N &\to X^\mathbb N \qquad
            \eta_\sigma(x_1,x_2, x_3, \dots) = (x_{\sigma^{-1}(1)},x_{\sigma^{-1}(2)}, x_{\sigma^{-1}(3)}, \dots).
            \end{align*} $\mu$ is called \textit{exchangeable} if, for every permutation $\sigma\colon \mathbb N \to \mathbb N$ which fixes all but a finite number of elements, we have that $(\eta_\sigma)_\ast \mu = \mu$. 
            \end{defn}  
        
            
            Let $X$ be a compact Hausdorff space. We define $\rad{X}$ to be the set of all
            Radon probability measures on the Borel $\sigma$-algebra on $X$, made into a compact Hausdorff space with the topology generated by the open sets $\{\mu \in \rad{X} \colon \int_{x\in X}f(x)\,\mathrm{d}\mu \in U\}$ for $U \subseteq \CC$ open, $f\colon X\to \CC$ continuous. 
            
            Given a Radon measure $\mu$ on $X$, by Kolmogorov extension, there is a Radon measure $\Tilde{\mu}$ on $X^\mathbb N$ defined on the basis of cylinder set opens by $\Tilde{\mu} \left(V\left[U_{n_1}, \dots, U_{n_k} \right] \right) = \prod_{i\in \mathbb N} \mu\left(U_{n_i}\right)$ for $U_{n_i} \subset X$ open. This is because for all $n \in \mathbb N$ there is a unique Radon measure on $X^n$, denoted by $\MeasPower{\mu}{n}$, which has $\MeasPower{\mu}{n}(U_1 \times \dots \times U_n) = \prod_{i\in \mathbb N} \mu\left(U_{i}\right)$ for $U_i \in \borel{X}$.
            
            \begin{thm}[Hewitt-Savage de Finetti Theorem \cite{Hewitt1955}] \label{thm: Hewitt Savage De Finetti Theorem}
                Let $\mu \in \rad{X^{\mathbb N}}$ be an exchangeable Radon probability measure on the countably infinite product of copies of $X$, with the Borel  $\sigma$-algebra. Then there exists a Radon probability measure $\nu$ on $\rad{X}$ (i.e.~$\nu\in \rad{\rad X}$) such that, for all measurable $U \subseteq X^{\mathbb N}$,
                $$
                \mu(U) = \int_{p \in \rad{X}} \tilde{p}(U) \, \mathrm{d} \nu.
                $$
            \end{thm}
            \begin{exmp}\label{cor:DeFinCoin}
                Supposing that we are modelling coin flips, so $X = \{ H, T \}$, then $ \rad X \cong [0,1]$, and the de Finetti result says that an exchangeable distribution on sequences of $H$ and $T$ can only come from picking from bag of coins with bias distributed according to some distribution on $[0,1]$ and then flipping the coin you have picked over and over forever.
            \end{exmp}
             
    \subsection{Rudiments of $\Cs$-algebras and Gelfand duality}\label{sec:cstar}
        \begin{defn}[C*-Algebra]
            An \emph{algebra} $V$ (over $\CC$) is a vector space $V$ over $\CC$ equipped with a binary operation of \emph{multiplication}, $\cdot \colon V \times V \to V$, which is bilinear. 
            If this multiplication is commutative, then $V$ is a \emph{commutative algebra}.
            We will assume that all algebras are \emph{unital}, i.e.~have a multiplicative unit.
            
            A \emph{Banach algebra} is an algebra $V$ equipped with a norm $\| \cdot \|$ such that $V$ is complete with respect to $\| \cdot \|$ and for all $x,y\in V$, $\| x\cdot y \| \leq \| x \| \| y \|$.
        %
            A \emph{$^\ast$-algebra} is a algebra $V$ that is equipped with an involution: a function $(-)^\ast \colon V \to V$ that is is self-inverse, a multiplication antihomomorphism (i.e. it reverses multiplication) and is conjugate linear.
            
            A bounded linear map between $^\ast$-algebras which preserves multiplication, the unit and involution is called a \emph{$^\ast$-homomorphism}.
        %
        \label{def: Cstar Algebra}
            A \emph{$\Cs$-algebra} $\CStA$ is a Banach $^\ast$-algebra such that for all $x \in \CStA$,
            $
            \|x^\ast x \| = \|x\|^2
            $.
            
            We write $\CstMIU$ for the category which has as its objects $\Cs$-algebras and $^\ast$-homomorphisms as its morphisms. It has as a full subcategory $\CCstMIU$ of commutative $\Cs$-algebras.
        \end{defn}
        
        \begin{exmp}
        For any Hilbert space $\HH$ over $\CC$, we denote the space of all bounded linear operators $\phi \colon \HH \to \HH$ by $\BB(\HH)$. $\BB(\HH)$ is the prototypical example of a $\Cs$-algebra. The generally non-commutative multiplication is given by composition of operators, the unit is the identity map and involution is taking the adjoint of a map. The norm is the operator norm.
        \end{exmp}

        \begin{thm}[e.g.~\cite{Landsman2017}, C.12]\label{thm:gelfandniamark}
        Every $\Cs$-algebra is isomorphic to a sub-algebra of $\BB(\HH)$ for some~$\HH$.
        \end{thm}
        
        \begin{exmp}[Commutative $\Cs$-algebras]
            One important example of a $\Cs$-algebra is the space $C(X) = \{ \psi \colon X \to \CC \, \vert \, \psi \text{ is continuous} \}$ for some compact Hausdorff space $X$, equipped with the (topological) supremum norm: $\| f\| = \sup_{x \in X} |f(x)| < \infty$. It is an algebra with multiplication and involution defined pointwise, and is commutative. The unit is the constant map to $1$.
        \end{exmp}
        
        This extends to a duality between commutative $\Cs$-algebras and the category~$\CH$ of compact Hausdorff spaces and continuous maps.
            
        \begin{thm}[Gelfand Duality] \label{thm: Gelfand Duality}
            The functor $\cont{-} = \Top(-,\CC) \colon \CH \to \op{\CCstMIU}$, which acts on morphisms by $\cont{f \colon X \to Y}\colon \phi \mapsto \phi \circ f$, is an equivalence of categories.
        \end{thm}
          \subsection{Positivity, Probabilistic Gelfand Duality, and the Radon Monad}
            \label{sec:probgelfand}
     
        \begin{defn}[Positivity in $\Cs$-algebras]\label{defn: Positivity in CStar algebras}
            For a $\Cs$-algebra $\CStA$, an element $x \in \CStA$ is called \emph{positive} if there is some $y \in \CStA$ such that $x = y^\ast y$. 
            
            In $\CC$, these are exactly the elements of the non-negative real line $\mathbb R_{0\leq}$. In $\BB(H)$, for some Hilbert space $H$, these are exactly the operators $\phi \colon H \to H$ such that for all $v \in H$, $\left\langle v \, \middle\vert \, \phi v\right\rangle \geq 0$. In $C(X)$, for some compact Hausdorff space $X$, these are the functions whose images lie exclusively in $\mathbb R_{0\leq}$.
            
            
            A linear map between $\Cs$-algebras, $f \colon \CStA_1 \to \CStA_2$, is called \emph{positive} if, for all $x \in \CStA_1$, $f(x^\ast x) \geq 0$. In other words, it takes positive elements of $\CStA_1$ to positive elements of $\CStA_2$. If both the domain and codomain have a unit, the map is called \emph{unital} if it takes the unit of the domain to the unit of the codomain.
            
            We will refer to the category of $\Cs$-algebras with positive, unital maps between them as $\CstPU$. It has $\CstMIU$ as a subcategory.
        \end{defn}
        
        Jacobs and Furber \cite{Furber2015} extended Gelfand duality stochastically: the addition of positive maps which are not multiplicative (i.e. not $^\ast$-homomorphisms) is equivalent to adding stochastic maps to $\CH$.
        \begin{defn}[Radon Monad, e.g.~\cite{Furber2015}]\label{def:radonA}
            Let $X$ be a compact Hausdorff space, and let $\rad{X}$ be the space of all Radon probability measures on $X$ as for Thm.\ref{thm: Hewitt Savage De Finetti Theorem} above. With this topology, $\rad{X}$ is both compact and Hausdorff.
            
            We regard $\radu$ as a monad on the category $\CH$ of compact Hausdorff spaces and continuous maps. The functor part is given by pushforward: let $f \colon X \to Y$ be a morphism in $\CH$; then $\rad{f}(\mu) := f_\ast(\mu)$.
            The unit of the monad takes $x \in X$ to the Dirac measure $\delta_x \in \rad{X}$, the distribution supported entirely at~$x$. The multiplication is a form of marginalization, or averaging:
            \begin{align*}
                \mathit{mult}:\radsqr{X} &\to \rad{X} \qquad
                \mathit{mult}(\phi)(U) := \int_{\mu \in \rad{X}} \mu(U) \, \mathrm{d}\phi(\mu)
            \end{align*}
            
        \end{defn}
        \begin{thm}[Probabilistic Gelfand Duality, \cite{Furber2015}] \label{thm: Stochastic Gelfand Duality}
            The functor $\cont{-} \colon \kl{\radu} \to \op{\CCstPU}$, which acts on morphisms by $\cont{f\colon X \to \rad{Y}} \colon \phi \mapsto \int \phi \, \mathrm{d} f(-)$, is an equivalence of categories between the Kleisli category of the Radon monad and the opposite of the category of commutative $\Cs$-algebras and positive unital maps.
        \end{thm}
            
    \subsection{States, State Spaces, and Convex Spaces}\label{sec:states}
        A particularly important class of positive, unital maps are states:
        \begin{defn}
        \label{defn: State on a Cstar Algebra}
            Positive, unital maps from a $\Cs$-algebra $\CStA$ to $\CC$ are called \emph{states on $\CStA$}.
        \end{defn}
        
        Under the probabilistic Gelfand duality (Thm.~\ref{thm: Stochastic Gelfand Duality}), for any compact Hausdorff space $X$ we have the correspondence (as sets)
        $$
            \CstPU(C(X), \CC) \ \ \cong\ \  \kl{\radu}\left(\left\{ \ast \right\}, X\right) \ \ = \ \ \CH\left(\left\{ \ast \right\}, \rad{X}\right)\ \  \cong\ \  \rad{X}.
        $$
        In fact, all of the objects in this correspondence have structures as convex, compact, Hausdorff spaces and are isomorphic as such. Thus it is meaningful to consider states on $\Cs$-algebras as a generalisation of classical probability distributions.
        
        \begin{exmp}[Density Matrices are States on a $\Cs$-algebra]\label{ex:density}
            Density matrices in quantum theory are given by operators $\rho \in \BB(\HH)$, for a finite dimensional Hilbert space $\HH$, with $\Tr{\rho} = 1$ such that for all $v \in \HH$, $\left\langle v \vert \rho v \right\rangle \geq 0$. For each such $\rho$, we may define a linear map $s_\rho \colon \BB(\HH) \to \CC$ by $s_\rho(a) = \Tr{\rho a}$. The trace condition says this map is unital. Further, in the finite dimensional case, where we can form a eigenvalue decomposition $\rho = \sum_{i\in I} p_i \outerprod{v_i}{v_i}$, then
            $$
                    s_\rho (a^\ast a) = \Tr{\rho a^\ast a} = \textstyle \sum_i p_i \innerprod{v_i}{a^\ast a v_i} = \sum_i p_i \innerprod{a v_i}{ a v_i} = \sum_i p_i \| a v_i \| \geq 0.
            $$
            So $s_\rho$ is a state on $\BB(\HH)$. In fact, in this finite dimensional case, all states on $\BB(\HH)$ are of this form.
        \end{exmp}
        
        So we can see $\Cs$-algebra states generalise both classical and quantum probability.
        
        \begin{defn}[State Space of a $\Cs$-algebra] \label{defn: State Space}
            Let $\CStA$ be a $\Cs$-algebra. \emph{The state space of $\CStA$}, denoted by $\stat{\CStA}:= \CstPU(\CStA, \CC)$, is the set of all states on $\CStA$ equipped with the coarsest topology such that for all $\CStElA \in \CStA$, the evaluation function $\mathrm{ev}_{\CStElA} \colon \stat{\CStA} \to \CC$ which takes $\rho \mapsto \rho(\CStElA)$ is continuous. The topology is generated by the sets $\mathrm{ev}^{-1}_{\CStElA}(\Omega)$ for $\CStElA \in \CStA$ and $\Omega \subset \CC$ open. $\stat{\CStA}$ has the additional properties of always being Hausdorff and compact.
            Moreover, 
            $\stat{-}$ extends to a functor $\op{\CstPU} \to \CH$ via $\stat{f} = - \circ f$.
        \end{defn}

        In fact $\stat{\CStA}$ also has an obvious convex structure where the convex combination $\lambda \rho_1 + (1-\lambda)\rho_2$ for $\lambda \in [0,1]$ and $\rho_1, \rho_2 \in \stat{\CStA}$ is evaluated pointwise using the addition of $\CC$. We now recall how this convex structure is functorial.
        
        \begin{defn}\label{defn: ConCom}
            For our purposes, a \emph{compact convex space} is a pair $(V,X)$ where $V$ is a convex, compact subset of a locally convex, Hausdorff topological vector space $X$. The category $\ConCom$ has as objects compact convex spaces $(V,X)$, and morphisms are affine, continuous maps between the convex subsets $V$ (only).
        \end{defn}
        
         \begin{thm}[\cite{Furber2015}]\label{thm: State Space Functor is Full and Faithful}
            The state space functor $\statu \colon \op{\CstPU} \to \CH$ factors through $\ConCom$. Moreover the resulting functor $\statu \colon \op{\CstPU} \to \ConCom$ is full and faithful.
        \end{thm}
        
        In other words, $\op{\CstPU}$ is isomorphic to a full subcategory of $\ConCom$ (characterized in \cite{Alfsen2001}).
       
        
        \begin{thm}[\cite{Furber2015}]\label{thm:concom-monadic}
            The category of Eilenberg-Moore algebras of the Radon monad, $\Alg{\radu}$ (Def.~\ref{def:radonA}), is equivalent to the category $\ConCom$ (Def.~\ref{defn: ConCom}).
        \end{thm}
        
        This realises the forgetful functor $U \colon \ConCom \to \CH$ as the right adjoint of the monadic adjunction $\radu' \dashv U$ for $\radu' \colon \CH \to \ConCom$ taking a space $X$ to the convex, compact Hausdorff space $\rad{X}$.
        
        Note then that all the objects of $\Alg{\radu}$ are either spaces of probability measures (the free algebras, equivalently objects of $\kl{\radu}$, Thm.~\ref{thm: Stochastic Gelfand Duality}), or they are (category-theoretical) ``quotients'' of these spaces (the non-free algebras). The transition from classical to quantum probability has as a crucial part the fact that some convex mixtures of outcomes to an experiment are equivalent (for example, different decompositions of mixed states into pure states).

    \subsection{Tensor products and complete positivity}
    \label{sec:tensors}
    Recall that the tensor product of Hilbert spaces, $\HH_1 \otimes_{H} \HH_2$, is the completion of the algebraic tensor product $\HH_1 \otimes \HH_2$ under the inner product norm. 
            
        
        \begin{defn}[Spatial Tensor Product of $\Cs$-algebras] \label{defn: Spatial Tensor Product}
            Let $\CStA_1$ and $\CStA_2$ be $\Cs$-algebras, with representations $\pi_i \colon \CStA_i \to \BB(\HH_i)$ (Thm.~\ref{thm:gelfandniamark}).
            Then the map $\pi_1\otimes \pi_2 \colon \CStA_1 \otimes \CStA_2 \to \BB(\HH_1\otimes_{H} \HH_2)$ given by $(\pi_1\otimes \pi_2) (A_1 \otimes A_2) = \pi_1(A_1)\otimes\pi_2(A_2)$ is a faithful representation of the vector space $\CStA_1 \otimes \CStA_2$ and thus we can use it to give a norm to that space, $\|a\|_{\ast} = \| (\pi_1\otimes\pi_2) (a)\|$, which is independent of the choice of representations ($\pi_i$). \textit{The spatial tensor product} $\CStA_1 \spotimes \CStA_2$ is the completion with respect to this norm. 
        \end{defn}
        
        There are other tensor products definable on $\Cs$-algebras, though this is the smallest. 
        %
        %
        %
        If a $\Cs$-algebra $\CStA$ is finite dimensional or commutative, then all possible $\Cs$-norms on $\CStA\otimes\CStB$ are equivalent.
        
        Not all positive maps are physical, and now that we have defined a tensor product on $\Cs$-algebras, we are able to define those that are.
        
        \begin{defn}[Completely Positive]\label{defn: Completely Positive}
            A linear map between $\Cs$-algebras $\phi \colon \CStA \to \CStB$ is called \emph{completely positive} if, for all $n \in \mathbb N$, the map
            $
            \phi \otimes 1_n \colon \CStA \otimes \BB(\CC^n) \to \CStB \otimes \BB(\CC^n)
            $
            is positive.
        \end{defn}
        
        All positive maps to or from commutative spaces are completely positive.
        
        \begin{exmp}Let $\HH,\KK$ be finite dimensional Hilbert spaces. Then a completely positive and unital map $\BB(\HH)\to \BB(\KK)$ induces a function
        $S(\BB(\KK))\to S(\BB(\HH))$ between the corresponding spaces of density matrices (Ex.~\ref{ex:density}); these are exactly the quantum channels (e.g.~\cite{NielsenChuang2000}). 
        \end{exmp}
        
        As we noticed a few times before, much of the seemingly infinite behaviour in different formulations of the De Finetti theorem is, in fact, the result of behaviour which happens on all possible finite truncations of a process, with requirements of consistency between them. 
        
        Given a $\Cs$-algebra $\CStA$, we can define  \raisebox{0pt}[0pt][0pt]{$\Xon{n} :=\underbrace{\CStA \spotimes \dots \spotimes \CStA}_{n \text{ times}}$}.
        For $n \leq m$, there is an isometric $^\ast$-homomorphism embedding of $\Cs$-algebras
        \begin{align*}
            \iota_{nm} \colon \Xon{n} &\to \Xon{m} \qquad
            \textstyle \iota_{nm}\Big(\bigotimes_{i=1}^n A_i\Big) \ := \  \bigotimes_{i=1}^n A_i \otimes \bigotimes_{i=n+1}^m 1_\CStA.
        \end{align*}
        
        \begin{defn}[Infinite Spatial Tensor Product]\label{defn: Infinite Spatial Tensor Product}
            The \emph{(countably) infinite spatial tensor product of $\CStA$} is defined as the colimit as Banach spaces (equivalently, in $\CstPU$) of the $\omega$-shaped diagram, for $\omega = \begin{tikzcd}[column sep =tiny] 0 \arrow[r] & 1 \ar[r] & 2 \ar[r] & 3 \ar[r]& \cdots \end{tikzcd}$, which has the objects $\Xon{n}$ for all $n \in \mathbb N$ and the morphisms $\iota_{nm}$ for all $n \leq m$. We will denote it by $\Xoinf$.
            
            This comes equipped with embeddings, again isometric $^\ast$-homomorphisms,
            $
            \psi_n \colon \Xon{n} \to \Xoinf
            $
            which we may intuitively imagine as taking $\bigotimes_{i=1}^n A_i \mapsto \bigotimes_{i=1}^n A_i \otimes \bigotimes_{i=n+1}^\infty 1_\CStA$ since for all $n \leq m$, $\psi_n = \psi_m \circ \iota_{nm}$.
            
            \begin{equation}\label{dgm:inftensor}
            \begin{tikzcd}[row sep=small]
                {\Xon{1}} & {\Xon{2}} & \dots & {\Xon{n}} & \dots & {\Xon{m}} & \dots \\
            	\\
            	&&& \Xoinf
            	\arrow["{\iota_{12}}", from=1-1, to=1-2]
            	\arrow["{\iota_{2n}}", from=1-2, to=1-4]
            	\arrow["{\iota_{nm}}", from=1-4, to=1-6]
            	\arrow["{\psi_1}"', dashed, to=3-4, from=1-1]
            	\arrow["{\psi_2}", dashed, to=3-4, from=1-2]
            	\arrow["{\psi_n}", dashed, to=3-4, from=1-4]
            	\arrow["{\psi_m}", dashed, to=3-4, from=1-6]
            \end{tikzcd}
            \end{equation}
        \end{defn}
        \subsection{The Quantum Kolmogorov Extension and de Finetti Theorems}\label{sec:prelim-qdf}
    
        \begin{thm}[Quantum Kolmogorov Extension Theorem \cite{Guichardet1969a}]\label{thm: Quantum Kolmogorov Extension Theorem}
            Let $\{ \rho_n\}_{n \in \mathbb N}$ be a sequence with $\rho_n \in \stat{\Xon{n}}$ such that, for all $n \leq m$, $\rho_n = \rho_m \circ \iota_{nm}$. There is a unique state $\rho \in S(\Xoinf)$ such that $\rho_n = \rho \circ \psi_n$ for all $n \in \mathbb N$ (with $\psi_n$ as in~\eqref{dgm:inftensor}).
        \end{thm}

        \begin{defn}[Exchangeable State]
            Given a permutation $\sigma \in \sym{n}$, we may define an associated map permuting the spaces of the n-fold tensor product $\CStA^{\otimes n}$
            \begin{align*}
                \eta_\sigma \colon \CStA^{\otimes n} &\to \CStA^{\otimes n} 
                \qquad \eta_\sigma(\CStElA_1 \otimes \dots \otimes \CStElA_n) \ := \ \CStElA_{\sigma^{-1}(1)} \otimes \dots \otimes \CStElA_{\sigma^{-1}(n)}
            \end{align*}
            
            A state $\rho_n \in \stat{\Xon{n}}$ is said to be \emph{symmetric} if, for all $\sigma \in \sym{n}$, $\rho_n = \rho_n\circ \eta_\sigma$. A state $\rho \colon \Xoinf \to \CC$ is said to be \emph{exchangeable} if, for all $n \in \mathbb N$, $\rho_n = \rho \circ \psi_n$ is symmetric. Note again that we are only concerned with permutations on a finite number of factors.
            
            \emph{The space of exchangeable states of $\Xoinf$}, denoted by $\sstat{\CStA}:= \{\rho \in S(\Xoinf) \, \vert \, \rho \text{ is exchangeable} \}$ is convex, compact and Hausdorff, so in $\ConCom$.
            
        \end{defn}

        \begin{thm}[Størmer's Quantum de Finetti Theorem]\label{thm: Quantum De Finetti Theorem}
            Let $\CStA$ be a $\Cs$-algebra. Then there is a bi-continuous, affine bijection $\sstat{\CStA} \cong \rad{\stat{\CStA}}$. In other words, the exchangeable state space of $\Xoinf$ is isomorphic to $\stat{\cont{\stat{\CStA}}}\cong\rad{\stat{\CStA}}$ in $\ConCom$.
        \end{thm}
        
        \begin{rmrk*}
            To be explicit about this isomorphism we note that given states $\rho_1 \in \stat{\CStA_1}, \rho_2 \in \stat{\CStA_2}$, we can form a unique state $\rho_1 \otimes \rho_2 \in \stat{\CStA_1 \spotimes \CStA_2}$ with the property that $\rho_1 \otimes \rho_2 (\CStElA_1 \otimes \CStElA_2) = \rho_1(\CStElA_1)\rho_2(\CStElA_2)$. Then, for $\rho \in \stat{\CStA}$, we define
            $$
            \rho^{\otimes n} := \underbrace{\rho \otimes \dots \otimes \rho}_{n \text{ times}}\in \stat{\Xon{n}}.
            $$
            There is a state on $\Xoinf$, $\rho^{\otimes \infty}\in\stat{\Xoinf}$,  via theorem \ref{thm: Quantum Kolmogorov Extension Theorem} from the sequence $\left\{ \rho^{\otimes n} \right\}_{n\in\mathbb N}$.
            
            This isomorphism then is given by $-\circ \Phi \colon  \stat{\cont{\stat{\CStA}}} \to \sstat{\CStA}$ where $\Phi \colon \Xoinf \to C(\stat{\CStA})$ is defined as $\Phi(\CStElA)(\rho) = \rho^{\otimes \infty} (\CStElA)$.
        \end{rmrk*}
\section{(Co)limits and the State Space Functor}
\label{sec:lemmas}
We now build on the prior work in Section~\ref{sec:prelim} to characterize the categorical limits in the categories of positive unital maps between $\Cs$-algebras ($\op{\CstPU}$) and compact convex spaces ($\ConCom$).

    \begin{lma} \label{thm: State Space Functor (Co)limit Reflection and Preservation}
        The state space functor $\statu \colon \op{\CstPU} \to \ConCom$ (Thm.~\ref{thm: State Space Functor is Full and Faithful}) preserves and reflects limits.
    \end{lma}
    
    That is to say that, given a diagram of $\Cs$-algebras $\CStA_{-} \colon \mathcal J \to \CstPU$, a cocone $\{ \CStA_j \to \CStB \}_{j \in \mathcal J}$ is a colimit if and only if the corresponding cone $\{ \stat{\CStB} \to \stat{\CStA_j} \}_{j \in \mathcal J}$ of the diagram $\op{\mathcal J} \to \op{\CstPU} \overset{\statu}{\to} \Alg{\mathcal R}$ is a limit.
    
    \begin{proof}
        Since $\statu$ is full and faithful (Thm.~\ref{thm: State Space Functor is Full and Faithful}), it must reflect limits and colimits. 
        
        There are monadic forgetful functors from $U \colon \Alg{\radu} \to \CH$ and from $ U' \colon \CH \to \Set$ (Thm.~\ref{thm:concom-monadic}, \cite[VI.9]{MacLane}). They create limits. The composition
        $$\begin{tikzcd}
            \op{\CstPU} \ar[r, "\statu"] &\Alg{\radu} \ar[r, "U"] & \CH \ar[r, "U'"] & \Set
        \end{tikzcd}$$
        is just the hom-functor $\CstPU(-,\mathbb C) \colon \op{\CstPU} \to \Set$. Thus, supposing a colimit exists in $\CstPU$, it is preserved via $\CstPU(-,\mathbb C)$ into a limit in $\Set$, which then creates a limit in $\CH$ and then creates another in $\Alg{\radu} \cong \ConCom$. Thus $\statu$ preserves the original colimit(/limit in $\op{\CstPU}$).
    \end{proof}
    
        
    
    The de Finetti limits we construct will be a simple result of the way that limits in $\ConCom$ are constructed \textit{pointwise} in the following way:
    
    \begin{lma}[Pointwise limits are limits in $\ConCom$]\label{thm: CCLcvx Pointwise Limit}
        Consider a diagram $\mathbf W_{-} \colon \mathcal J \to \ConCom$. That is, consider a collection $\{ \mathbf W_i\}_{i \in J}$ of compact convex spaces with affine continuous maps ${\{\xi_k \colon  \mathbf W_i \to \mathbf W_j \}_{k \in \mathcal J(i,j)}}$. 
        Let $\mathbf W$ be a convex, compact Hausdorff space with morphisms $\omega_i \colon \mathbf W \to \mathbf W_i$ satisfying the following properties:
        
        \begin{itemize}
            \item $\{\omega_i \colon \mathbf W \to \mathbf W_i \}_{i \in J}$ is a cone over the diagram: for all $\xi_k \colon \mathbf W_i \to \mathbf W_j$ in the diagram, $\omega_i = \omega_j \circ \xi_k$.
            \item For any collection of elements $( w_i )_{i \in J}\in \prod_{i\in J} W_i$ which are compatible, in the sense that for all $\xi_k \colon \mathbf W_j \to \mathbf W_i$ in the diagram, $w_i = \xi_k(w_k)$, there is a unique element $w \in \mathbf W$ with
        $
        w_i = \omega_i(w)$.
        \end{itemize}

        Then $\mathbf W$ is the limit of the diagram in $\CH$ and $\ConCom$.
    \end{lma}
    
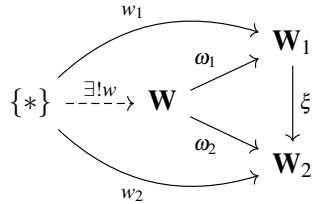
    $$
    \begin{tikzcd}[row sep=tiny]
                            &           & \mathbf{W}_1  \\
        \{ \ast \}          & \mathbf{W} &    \\
                            &           & \mathbf{W}_2
        \ar["w_1", from=2-1, to=1-3, bend left] \ar["\omega_1", from=2-2, to=1-3]
        \ar["\exists ! w", from=2-1, to=2-2, dashed]
        \ar["w_2"', from=2-1, to=3-3, bend right] \ar["\xi", from=1-3, to=3-3] \ar["\omega_2"', from=2-2, to=3-3]
    \end{tikzcd} 
    $$

    If the cone comprises state spaces $\mathbf W_j = \stat{\CStA_j}$ and $\mathbf W = \stat{\CStB}$, rather than general objects of $\ConCom$, then the second condition can alternatively be visualised on the level of positive unital maps:
    $$
    
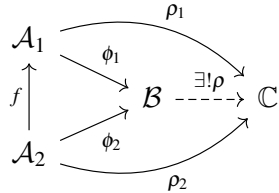
\begin{tikzcd}[row sep=tiny]
        \CStA_1    &           &               \\
                        & \CStB & \mathbb C    \\
        \CStA_2    &           &
        \ar["\rho_1", from=1-1, to=2-3, bend left] \ar[from=1-1, to=2-2, "\phi_1"]
        \ar["\exists ! \rho", from=2-2, to=2-3, dashed]
        \ar["\rho_2"', from=3-1, to=2-3, bend right] \ar["f", from=3-1, to=1-1] \ar["\phi_2"', from=3-1, to=2-2]
    \end{tikzcd} 
    $$
    
    In this way we may talk about limits of state spaces being built \textit{pointwise}. If every set of compatible states corresponds to a limiting state on some other space, then this is naturally a limit, without having to concern outselves with continuity or affineness of the limiting maps.
    
    \begin{proof}
        The forgetful functors $U\colon \ConCom \to \CH$ and $U' \colon \CH \to \Set$ are both monadic (Thm.~\ref{thm:concom-monadic}, \cite[VI.9]{MacLane}) and thus create limits. 
        
        Suppose we have a diagram $\mathbf W_{-} \colon \mathcal J \to \ConCom$ as above, and a pointwise limit $\mathbf W$ with maps $\omega_i \colon \mathbf W \to \mathbf W_i$, where $W_i$ is the underlying set for each $\mathbf W_i$. Then for any set $A$, a cone of the diagram of $W_i$s, that is a collection of maps $\left\{ f_i \colon A \to W_i \right\}_{i \in I}$ gives a unique map $f \colon A \to W$ by, for each $a \in A$, letting $f(a)$ be the element of $W$ corresponding to the collection $\{ f_i(a) \in \mathbf W_i \}$ under the pointwise limit property. Thus, $W$ is the limit in $\Set$ of the diagram
        $$
        \begin{tikzcd}
            \mathcal J & \ConCom & \CH & \Set
            \ar["\mathbf W_{-}", from=1-1, to=1-2] \ar["U", from=1-2, to=1-3] \ar["U'", from=1-3, to=1-4]
        \end{tikzcd}
        $$

        Since $\CH$ is monadic over $\Set$, then $\mathbf W$, regarded as a compact Hausdorff space, is created as the limit in $\CH$ and similarly by monadicity of $\ConCom$ over $\CH$, it is created as the limit in $\ConCom$.
    \end{proof}
    
    Informally, what this lemma tells us is that non-categorical, state-based treatments of limiting behaviours have been categorical all along. Looking at these things pointwise is fine because the structural properties take care of themselves.
\section{A Quantum de Finetti Theorem as a Categorical Limit}\label{sec:qdf}
We now use the lemmas of Section~\ref{sec:lemmas} to recast the quantum Kolmogorov extension theorem~(Thm.~\ref{thm: Quantum Kolmogorov Extension Theorem}) and quantum de Finetti theorem (Thm.~\ref{thm: Quantum De Finetti Theorem}) in a categorical light (Thms.~\ref{thm:catqkolmogorov}, \ref{thm:qdf}). 
\begin{thm}[Quantum Kolmogorov Extension Theorem as a Categorical Limit]\label{thm:catqkolmogorov}
    Let $\CStA$ be a $\Cs$-algebra. The limit of the diagram of state spaces (as compact convex spaces in  $\ConCom$)
    \begin{equation}\label{eqn: Kolmogorov Quantum Diagram}
    \begin{tikzcd}
        \stat{\CStA}    & \stat{\Xon{2}} \ar[l, "- \circ \iota_{12}"'] & \stat{\Xon{3}} \ar[l, "- \circ \iota_{23}"'] & \cdots \ar[l]
    \end{tikzcd}
    \end{equation}
    is $\stat{\Xoinf}$ with the inclusions $- \circ \psi_n \colon \stat{\Xoinf} \to \stat{\Xon{n}}$.
\end{thm}

\begin{proof}
    Since this diagram is exactly the image under $\statu$ of that which has $\Xoinf$ as its colimit (Def.~\ref{defn: Infinite Spatial Tensor Product}), and by Lma.~\ref{thm: State Space Functor (Co)limit Reflection and Preservation}, $\statu$ preserves colimits.
\end{proof}

\begin{lma}[Quantum de Finetti Theorem as a Categorical Limit For Positive Maps]\label{lemma:qdf}
    Let $\CStA$ be a $\Cs$-algebra. Let $\Iinj$ be the category of the finite sets $\{ 1 ,\dots, n\}$ for $n \in \mathbb N$ and injections between them. We define an $\Iinj$-indexed diagram in $\CstPU$ by taking $\{ 1 ,\dots, n\}$ to $\Xon{n}$ and, for $n \leq m$ and an injection $\tau \colon \{1,\dots,n\} \hookrightarrow \{ 1, \dots, m \}$, we define
    \begin{equation}\label{eqn:QdF-diagram}
        \eta_\tau \colon \CStA^{\otimes n} \to \CStA^{\otimes m} \\
    \end{equation}
    by taking $A_1 \otimes \dots \otimes A_n$ to the element $B_1 \otimes \dots \otimes B_m$ which has 
    $$
        B_j = \begin{cases} A_i & \text{if } j = \tau (i),\\ 1 & \text{otherwise}. \end{cases}
    $$
    The colimit of this diagram in $\CstPU$ is $\cont{\stat{\CStA}}$. 
\end{lma}

In other words, the space of exchangeable sequences of states is a limit of a diagram and is affinely isomorphic to the space of Radon probability measures on $S(\CStA)$.

\begin{proof}
    Størmer's proof of Thm.~\ref{thm: Quantum De Finetti Theorem}~\cite{Stormer1969} gives a bi-continuous, affine bijection $\sstat{\CStA} \cong \rad{\stat{\CStA}}$. In other words, the symmetric state space of $\Xoinf$ is isomorphic to $\rad{\stat{\CStA}}$ in $\ConCom$. All that is necessary then is to show that $\sstat{\CStA} \subset \stat{\Xoinf}$, with the morphisms $\rho \mapsto \rho \circ \psi_n$ (with $\psi_n$ as in \eqref{dgm:inftensor}) is the limit of the diagram above.
    We do this using Lma.\ref{thm: CCLcvx Pointwise Limit}.
    
    Via Thm.~\ref{thm: State Space Functor is Full and Faithful}, we can regard the diagram $\{\eta_\tau \colon \CStA^{\otimes n} \to \CStA^{\otimes m}\}$ of C$^*$-algebras as a diagram of state spaces $\{S(\eta_\tau) \colon S(\CStA^{\otimes m}) \to S(\CStA^{\otimes n})\}$ in $\ConCom$.
    Since any $\rho \in \sstat{\CStA}$ is symmetric we get that $\rho \mapsto \rho \circ \psi_n$  is a cone for $\{S(\eta_\tau)\}$.
    
     Any collection of states $\{ \rho_n \in \stat{\Xon{n}} \}_{n \in \mathbb N}$ which satisfies the diagram $\{S(\eta_\tau)\}$ also satisfies the diagram \ref{eqn: Kolmogorov Quantum Diagram} and thus we get a state on $\Xoinf$, $\rho \in \stat{\Xoinf}$, such that $\rho_n = \rho \circ \psi_n$ for all $n \in \mathbb N$. Since for any permutation $\sigma$ of $\{1, \dots, n\}$, $\rho_n \circ \eta_\sigma = \rho_n$, $\rho$ is symmetric. Thus, $\rho \in \sstat{\CStA}$ and $\sstat{\CStA}$ is a pointwise limit, and thus a limit for the diagram $\{S(\eta_\tau)\}$ in $\ConCom$, via Lma.~\ref{thm: CCLcvx Pointwise Limit}.
    
    That $\cont{\stat{\CStA}}$ is the colimit in the diagram in $\CstPU$ is a corollary of Lma.\ref{thm: State Space Functor (Co)limit Reflection and Preservation}: $\statu$ reflects limits.
\end{proof}

\begin{thm}[Quantum de Finetti Theorem as a Categorical Limit For Quantum Channels]\label{thm:qdf}
    The colimit of the diagram~\eqref{eqn:QdF-diagram} in $\CstCPU$ is $\cont{\stat{\CStA}}$.
\end{thm}

\begin{proof}
     Since all the maps in the diagram~\eqref{eqn:QdF-diagram} are $^\ast$-homomorphisms, in particular completely positive, the diagram factors through $\CstCPU$ and, because $\cont{\stat{\CStA}}$ is commutative, the colimiting maps are always completely positive. Thus from~Lma.~\ref{lemma:qdf}, $\cont{\stat{\CStA}}$ is also the colimit in $\CstCPU$.
\end{proof}
\section{Illustrations}\label{sec:illustration}
\newcommand{\ket}[1]{\ensuremath{\left|#1\right\rangle}} 

To illustrate the de Finetti construction, we focus temporarily on states of a qubit (i.e.~$\CStA=\BB(\CC^2)$). Our categorical quantum de Finetti theorem gives a universal property to the infinite dimensional space of all states ($C(S(\CStA))$), in terms of finite dimensional spaces ($\CC^2,\CC^4$ etc..). We can thus use conventional methods from finite dimensional categorical/diagrammatic quantum mechanics to analyze $C(S(\CStA))$. 

Since every state in 
$\BB(\CC^{(2^n)})$ corresponds to a quantum circuit with $n$ output qubits,
we can describe a sequence of states by giving a sequence of circuits.
For example, the following sequence of circuits describes a sequence of qubit states that is exchangeable (in that postcomposition with any permutation or discarding of the qubit wires respects the sequence).
For familiarity, we use a quantum circuit notation, but any diagrammatic notation with discarding could be used (e.g.~\cite{Carette2019}).

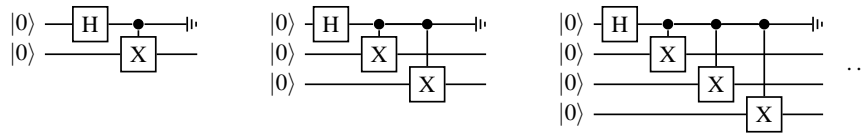
\usetikzlibrary{fit,decorations.pathreplacing,circuits.ee.IEC}  

\begin{align*}&\scalebox{0.8}{\begin{tikzpicture}[thick,circuit ee IEC,yscale=0.5,xscale=0.8]
    \tikzstyle{operator} = [draw,fill=white,minimum size=1.5em] 
    \tikzstyle{phase} = [fill,shape=circle,minimum size=5pt,inner sep=0pt]
    %
    \node at (-0.4,0) (q1a) {\ket{0}};
    \node at (-0.4,-1) (q2a) {\ket{0}};
    %
    %
    \node[operator] (op10a) at (1,0) {H} edge [-] (q1a);
    \node[phase] (op11a) at (2,0) {} edge [-] (op10a);
    \node[operator] (op12a) at (2,-1) {X} edge [-] (q2a);
    \draw[-] (op11a) -- (op12a);
    %
    \node[ground,right] (o1a)  at (3,0) {} edge [-] (op11a);
    \node (o2a)  at (3.4,-1) {} edge [-] (op12a);
    %
    \node at (5,0) (q1b) {\ket{0}};
    \node at (5,-1) (q2b) {\ket{0}};
    \node at (5,-2) (q3b) {\ket{0}};
    %
    \node[operator] (op10b) at (6,0) {H} edge [-] (q1b);
    %
    \node[phase] (op11b) at (7,0) {} edge [-] (op10b);
    \node[operator] (op12b) at (7,-1) {X} edge [-] (q2b);
    \draw[-] (op11b) -- (op12b);
    %
    \node[phase] (op21b) at (8,0) {} edge [-] (op11b);
    \node[operator] (op22b) at (8,-2) {X} edge [-] (q3b);
    \draw[-] (op21b) -- (op22b);
    %
    \node[ground,right] (o1b)  at (9,0) {} edge [-] (op11b);
    \node (o2b)  at (9.4,-1) {} edge [-] (op12b);
    \node (o3b)  at (9.4,-2) {} edge [-] (op22b);
    %
    \node at (11,0) (q1c) {\ket{0}};
    \node at (11,-1) (q2c) {\ket{0}};
    \node at (11,-2) (q3c) {\ket{0}};
    \node at (11,-3) (q4c) {\ket{0}};
    %
    \node[operator] (op10c) at (12,0) {H} edge [-] (q1c);
    %
    \node[phase] (op11c) at (13,0) {} edge [-] (op10c);
    \node[operator] (op12c) at (13,-1) {X} edge [-] (q2c);
    \draw[-] (op11c) -- (op12c);
    %
    \node[phase] (op21c) at (14,0) {} edge [-] (op11c);
    \node[operator] (op22c) at (14,-2) {X} edge [-] (q3c);
    \draw[-] (op21c) -- (op22c);
    %
    \node[phase] (op31c) at (15,0) {} edge [-] (op21c);
    \node[operator] (op34c) at (15,-3) {X} edge [-] (q4c);
    \draw[-] (op31c) -- (op34c);
    %
    \node[ground,right] (o1c)  at (16,0) {} edge [-] (op11c);
    \node (o2c)  at (16.4,-1) {} edge [-] (op12c);
    \node (o3c)  at (16.4,-2) {} edge [-] (op22c);
    \node (o4c)  at (16.4,-3) {} edge [-] (op34c);
    \node (dots) at (17,-1.5) {\dots};
    \end{tikzpicture}}
                \intertext{
                This exchangeable sequence corresponds to a belief about a qubit state: that the state is pure, and is either $\ket 0$ (with probability $0.5$) or $\ket 1$ (with probability $0.5$). The following sequence is also exchangeable:}
  &\scalebox{0.8}
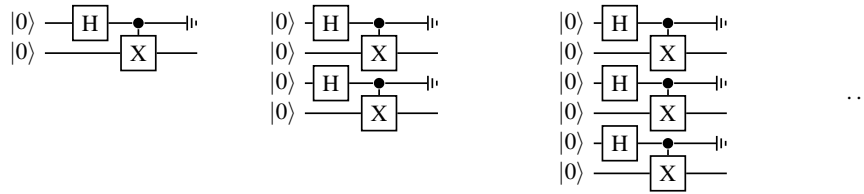
{\begin{tikzpicture}[thick,circuit ee IEC,yscale=0.5,xscale=0.8]
    \tikzstyle{operator} = [draw,fill=white,minimum size=1.5em] 
    \tikzstyle{phase} = [fill,shape=circle,minimum size=5pt,inner sep=0pt]
    %
    \node at (-0.4,0) (q1a) {\ket{0}};
    \node at (-0.4,-1) (q2a) {\ket{0}};
    %
    %
    \node[operator] (op10a) at (1,0) {H} edge [-] (q1a);
    \node[phase] (op11a) at (2,0) {} edge [-] (op10a);
    \node[operator] (op12a) at (2,-1) {X} edge [-] (q2a);
    \draw[-] (op11a) -- (op12a);
    %
    \node[ground,right] (o1a)  at (3,0) {} edge [-] (op11a);
    \node (o2a)  at (3.4,-1) {} edge [-] (op12a);
    %
    \node at (5,0) (q1b) {\ket{0}};
    \node at (5,-1) (q2b) {\ket{0}};
    \node at (5,-2) (q3b) {\ket{0}};
    \node at (5,-3) (q4b) {\ket{0}};
    %
    \node[operator] (op01b) at (6,0) {H} edge [-] (q1b);
    \node[operator] (op03b) at (6,-2) {H} edge [-] (q3b);
    %
    \node[phase] (op11b) at (7,0) {} edge [-] (op01b);
    \node[operator] (op12b) at (7,-1) {X} edge [-] (q2b);
    \draw[-] (op11b) -- (op12b);
    \node[phase] (op13b) at (7,-2) {} edge [-] (op03b);
    \node[operator] (op14b) at (7,-3) {X} edge [-] (q4b);
    \draw[-] (op13b) -- (op14b);
    \node[ground,right] (o1b)  at (8,0) {} edge [-] (op11b);
    \node[ground,right] (o3b)  at (8,-2) {} edge [-] (op13b);
    \node (o2b)  at (8.4,-1) {} edge [-] (op12b);
    \node (o4b)  at (8.4,-3) {} edge [-] (op14b);
    %
    \node at (11,0) (q1c) {\ket{0}};
    \node at (11,-1) (q2c) {\ket{0}};
    \node at (11,-2) (q3c) {\ket{0}};
    \node at (11,-3) (q4c) {\ket{0}};
    \node at (11,-4) (q5c) {\ket{0}};
    \node at (11,-5) (q6c) {\ket{0}};
    %
    \node[operator] (op01c) at (12,0) {H} edge [-] (q1c);
    \node[operator] (op03c) at (12,-2) {H} edge [-] (q3c);
    \node[operator] (op05c) at (12,-4) {H} edge [-] (q5c);
    %
    \node[phase] (op11c) at (13,0) {} edge [-] (op01c);
    \node[operator] (op12c) at (13,-1) {X} edge [-] (q2c);
    \draw[-] (op11c) -- (op12c);
    \node[phase] (op13c) at (13,-2) {} edge [-] (op03c);
    \node[operator] (op14c) at (13,-3) {X} edge [-] (q4c);
    \draw[-] (op13c) -- (op14c);
    \node[phase] (op15c) at (13,-4) {} edge [-] (op05c);
    \node[operator] (op16c) at (13,-5) {X} edge [-] (q6c);
    \draw[-] (op15c) -- (op16c);
    %
    \node[ground,right] (o1c)  at (14,0) {} edge [-] (op11c);
    \node[ground,right] (o3c)  at (14,-2) {} edge [-] (op13c);
    \node[ground,right] (o5c)  at (14,-4) {} edge [-] (op15c);
    \node (o2c)  at (14.4,-1) {} edge [-] (op12c);
    \node (o4c)  at (14.4,-3) {} edge [-] (op14c);
    \node (o6c)  at (14.4,-5) {} edge [-] (op16c);
    \node (dots) at (17,-2.5) {\dots};
    \end{tikzpicture}}\end{align*}
    But this sequence corresponds to a different belief about a qubit state: that the state is definitely the mixed state with density matrix $\frac 12 (\begin{smallmatrix}1 & 0 \\ 0 & 1\end{smallmatrix})$. In each case, the overall sequence of states is equivalent to the process of  drawing a state at random from a classical distribution and then repeating it.
    
    There are many other exchangeable sequences of qubit states. For example, there is one corresponding to the belief that a qubit is in a pure state somewhere on the equator of the Bloch sphere, $\frac 1{\sqrt 2}(\ket 0+e^{i\phi} \ket 1)$, with $\phi$ uniformly distributed in $[0,2\pi]$. There is also one corresponding to the belief that a qubit is in a totally unknown state, i.e.~density matrix $\frac{r^{\frac 13}}2 (\begin{smallmatrix}1+z&\sqrt{1-z^2}e^{-i\theta}\\\sqrt{1-z^2}e^{i\theta}&1-z\end{smallmatrix})$ with $r,z,\theta$ uniformly distributed in $[0,1]$, $[-1,1]$ and $[0,2\pi]$ respectively. 
    
    Our categorical version of the quantum de Finetti theorem allows us to also consider sequences with parameters.  
    Since every completely positive unital map 
    $\BB(\CC^{(2^n)})\to \BB(\CC^{(2^m)})$ corresponds to a quantum circuit with $n$ output qubits and $m$ input qubits,
    we can describe a cone with apex $\BB(\CC^{(2^m)})$ by giving a sequence of circuits with $m$ inputs.

For example, the following sequence of circuits describes a cone with apex $\BB(\CC^2)$. 
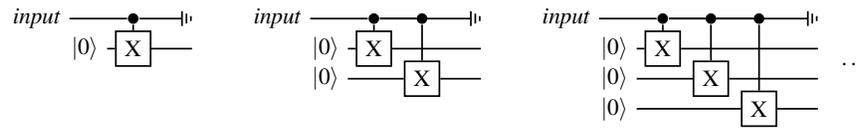
\[\scalebox{0.8}{\begin{tikzpicture}[thick,circuit ee IEC,yscale=0.5,xscale=0.8]
    \tikzstyle{operator} = [draw,fill=white,minimum size=1.5em] 
    \tikzstyle{phase} = [fill,shape=circle,minimum size=5pt,inner sep=0pt]
    %
    \node at (0,0) (q1a) {\textit{input}};
    \node at (1,-1) (q2a) {\ket{0}};
    \node[phase] (op11a) at (2,0) {} edge [-] (q1a);
    \node[operator] (op12a) at (2,-1) {X} edge [-] (q2a);
    \draw[-] (op11a) -- (op12a);
    %
    \node[ground,right] (o1a)  at (3,0) {} edge [-] (op11a);
    \node (o2a)  at (3.4,-1) {} edge [-] (op12a);
    %
    \node at (5,0) (q1b) {\textit{input}};
    \node at (6,-1) (q2b) {\ket{0}};
    \node at (6,-2) (q3b) {\ket{0}};
    %
    \node[phase] (op11b) at (7,0) {} edge [-] (q1b);
    \node[operator] (op12b) at (7,-1) {X} edge [-] (q2b);
    \draw[-] (op11b) -- (op12b);
    %
    \node[phase] (op21b) at (8,0) {} edge [-] (op11b);
    \node[operator] (op22b) at (8,-2) {X} edge [-] (q3b);
    \draw[-] (op21b) -- (op22b);
    %
    \node[ground,right] (o1b)  at (9,0) {} edge [-] (op11b);
    \node (o2b)  at (9.4,-1) {} edge [-] (op12b);
    \node (o3b)  at (9.4,-2) {} edge [-] (op22b);
    %
    \node at (11,0) (q1c) {\textit{input}};
    \node at (12,-1) (q2c) {\ket{0}};
    \node at (12,-2) (q3c) {\ket{0}};
    \node at (12,-3) (q4c) {\ket{0}};
    %
    \node[phase] (op11c) at (13,0) {} edge [-] (q1c);
    \node[operator] (op12c) at (13,-1) {X} edge [-] (q2c);
    \draw[-] (op11c) -- (op12c);
    %
    \node[phase] (op21c) at (14,0) {} edge [-] (op11c);
    \node[operator] (op22c) at (14,-2) {X} edge [-] (q3c);
    \draw[-] (op21c) -- (op22c);
    %
    \node[phase] (op31c) at (15,0) {} edge [-] (op21c);
    \node[operator] (op34c) at (15,-3) {X} edge [-] (q4c);
    \draw[-] (op31c) -- (op34c);
    %
    \node[ground,right] (o1c)  at (16,0) {} edge [-] (op11c);
    \node (o2c)  at (16.4,-1) {} edge [-] (op12c);
    \node (o3c)  at (16.4,-2) {} edge [-] (op22c);
    \node (o4c)  at (16.4,-3) {} edge [-] (op34c);
    \node (dots) at (17,-1.5) {\dots};
    \end{tikzpicture}}
\]
This corresponds to a belief that the quantum state is pure, and either $\ket 0$ or $\ket 1$, with the probability determined by a standard basis measurement of the input state. Indeed, the categorical quantum de Finetti theorem (Thm.~\ref{thm:qdf}) says that every exchangeable sequence of circuits is equivalent to a sequence where each circuit first measures all the input qubits, resulting in random classical data, and then generates a quantum state depending on this classical outcome. 

\section{Aside on the Hewitt-Savage de Finetti Theorem as a Categorical Limit}
As an aside, we note that de Finetti theorem for classical probability (Thm.~\ref{thm: Hewitt Savage De Finetti Theorem}) now also arises  as a categorical limit. We express it first in the category $\ConCom$ of compact convex spaces, but then state it in terms of the Kleisli category of the Radon monad to show the similarity with the previous result in~\cite{JacobsStaton2020}.

\begin{thm}[Categorical Hewitt-Savage De Finetti Theorem]\label{thm: Categorical Hewitt-Savage De Finetti Theorem}
                Let $X$ be a compact Hausdorff space. Consider the diagram $\op{\Iinj} \to \ConCom$ which takes $\{1 , \dots, n\}$ to $\rad{X^n}$ and an injective function $\tau \colon  \{1 , \dots, n\} \to \{ 1, \dots , m \}$ to $\zeta_\tau \colon \rad{X^m} \to \rad{X^n}$, defined by $\zeta_\tau(\mu)(A) = \mu(\Tilde A_\tau)$ for $$
                \Tilde A_\tau := \left\{ \left(x_1 ,\dots, x_m\right) \in X^m \, \vert \, \left(x_{\tau(1)},x_{\tau(2)},\dots, x_{\tau(n)}\right)\in A \right\}.
                $$
                The limit of this diagram is the space of exchangeable measures on $X^\mathbb N$, and is isomorphic to $\rad{\rad{X}}$, where the maps $\rad{\rad{X}} \to \rad{X^n}$ take a measure $\Phi$ on $\rad{X}$ to the measure
                $$
                    A \in \borel{X^n} \mapsto \int_{\mu \in \rad{X}} \underbrace{\mu\times \dots \times \mu}_{n \text{ times}}(A) \, \mathrm{d}\Phi.
                $$
\end{thm}

\begin{proof}
    This follows from instantiating Theorem~\ref{thm:qdf} with the commutative $\Cs$-algebra $\cont{X}$ and noting that $\cont{X} \spotimes \cont{Y}\cong \cont{X \times Y}$.
  \end{proof}

  To emphasise the connection with~\cite{JacobsStaton2020}, we write $X\klto Y$ for a Kleisli morphism $X\to \radu(Y)$. 

\begin{cor}[Categorical de Finetti Theorem in $\kl{\radu}$]
    For some $X \in \CH$, consider the diagram $\op{\Iinj} \to \kl{\radu}$ into the Kleisli category of the Radon monad, which takes $\{1 , \dots, n\}$ to $X^n$ and each injection $\tau \colon  \{1 , \dots, n\} \to \{ 1, \dots , m \}$ to the Kleisli-map $\eta_\tau \colon X^m \klto X^n$ given by
    $$
        \eta_\tau (x_1, \dots, x_m) = \delta_{\left(x_{\tau(1)},x_{\tau(2)},\dots, x_{\tau(n)}\right)}
    $$
    This diagram has limit $\rad{X}$ in $\kl{\radu}$, with maps $\iid{n} \colon \rad{X} \klto X^n$ given, for measurable $A \subset X^n$, by
    $$
        \iid{n}(\mu)(A) = (\underbrace{\mu\times \dots \times \mu}_{n \text{ times}})(A).
    $$
\end{cor}

\begin{proof}
    The limit in theorem~\ref{thm: Categorical Hewitt-Savage De Finetti Theorem} is reflected into $\kl{\radu}$, as this category is a full subcategory of $\Alg{\radu} \cong \ConCom$ and the limit itself is a free algebra.
\end{proof}

The notation $\iid{n}$ arises because $\underbrace{\mu\times \dots \times \mu}_{n \text{ times}}$ describes independent and identical distributions $\mu$ on $n$ copies of $X$. 
\section{Concluding remarks}

We have shown that the quantum de Finetti theorem amounts to a categorical limit for a diagram $\CStA^{\otimes-}:\op\Iinj\to \op\CstCPU$ (Theorem~\ref{thm:qdf}). 
This puts the quantum de Finetti theorem, a cornerstone of quantum Bayesianism and a starting point for quantum tomography, in the setting of categorical and diagrammatic quantum theory. We have focused on $\Cs$-algebras, but the set-up is relevant more broadly. Recall that an \emph{affine monoidal category} is a symmetric monoidal category with a terminal unit, and these are argued to be causal models of quantum theory (e.g.~\cite{Coecke2014}). Now, $\op\Iinj$ is the \emph{free} affine monoidal category on one generator (e.g.~\cite{MalherbeScottSelinger2013}), and so $\CStA^{\otimes-}:\op\Iinj\to \op{\CstCPU}$ is canonical for $\CStA$. So we can consider de Finetti limits in other models of quantum theory, perhaps making a bridge to the test space analysis of~\cite{BarrettLeifer2009}.

\paragraph{Acknowledgements.}
One starting point for this work was a discussion at the 2020 Workshop on Categorical Probability and Statistics, and we thank the organizers and participants at that workshop, in particular Robert Furber, Bart Jacobs, Arthur Parzygnat and Robert Spekkens. Thanks too to the Oxford group and anonymous reviewers. 

Research supported by AFOSR award number FA9550-21-1-0038; the
ERC BLAST grant; and a Royal Society University Research Fellowship.

\bibliographystyle{eptcs}
\bibliography{generic}

\end{document}